\documentclass{article}

\usepackage{arxiv}

\usepackage[utf8]{inputenc} 
\usepackage[T1]{fontenc}    
\usepackage{hyperref}       
\usepackage{url}            
\usepackage{booktabs}       
\usepackage{amsfonts}       
\usepackage{amssymb}
\usepackage{nicefrac}       
\usepackage{microtype}      
\usepackage{cleveref}       
\usepackage{lipsum}         
\usepackage{graphicx}
\usepackage{doi}
\usepackage{algorithm}
\usepackage{algorithmicx}
\usepackage{algpseudocode}

\title{Directional movement of a collective of compassless automata on square lattice of width 2}

\newif\ifuniqueAffiliation
\uniqueAffiliationtrue

\ifuniqueAffiliation 
\author{ \href{https://orcid.org/0000-0001-6435-819X}{\includegraphics[scale=0.06]{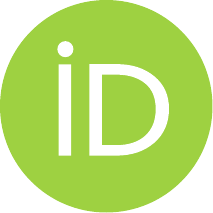}\hspace{1mm}Sergey V.~Sapunov} \\
	Institute of Applied Mathematics and Mechanics\\
	National Academy of Sciences of Ukraine\\
	Batiouk Str. 19, 84116 Sloviansk, Ukraine\\
	\texttt{sapunov@nas.gov.ua} \\
}
\fi

\renewcommand{\shorttitle}{Directional movement of a collective of compassless automata on square lattice}

\newtheorem{theorem}{Theorem}
\newenvironment{proof}[1][Proof]{\begin{trivlist}
\item[\hskip \labelsep {\bfseries #1}]}{\end{trivlist}}
\newcommand{\qed}{\nobreak \ifvmode \relax \else
      \ifdim\lastskip<1.5em \hskip-\lastskip
      \hskip1.5em plus0em minus0.5em \fi \nobreak
      \vrule height0.75em width0.5em depth0.25em\fi}

\begin{document}
\maketitle

\begin{abstract}
We study the following problem: Can a collective of finite automata maintain directed movement on a two-dimensional integer lattice of width 2, where the elements (vertices) are anonymous? 
The automata do not distinguish between vertices based on their coordinates of direction (that means each automaton has no compass).
We considered collectives consisting of an automaton and some pebbles, which are automata of the simplest form, whose positions are entirely determined by the automaton.
We demonstrate that a collective of one automaton and a maximum of three pebbles cannot maintain a direction of movement on the lattice. 
However, a collective of one automaton and four pebbles can do so.
\end{abstract}


\section{Introduction}
Automata walking on graphs are a mathematical formalization of autonomous mobile agents with limited memory operating in discrete environments.
Studies on the behavior of automata in finite and infinite labyrinths, which are embedded directed graphs of a specific form, have emerged and are intensively developing under this model \cite{Hemmerling1989, Kilibarda2003, Kilibarda2003a}.
Research in this area has various applications, such as image analysis \cite{Kari2006, Stamatovic2017} and mobile robotics navigation \cite{Dudek2010}.
The assumption that automata operating in labyrinths can distinguish directions, i.e., they have a compass, is crucial for the results on automata and labyrinths \cite{Blum1978, Donald2012}.

This paper aims to investigate collectives of compassless automata that interact in an environment represented by an infinite two-dimensional integer lattice of width $2$.
Each automaton receives information about the presence or absence of other automata at the neighboring vertices (lattice points) when interacting with the environment. Based on this input, it moves to one of these vertices.
The automata do not differentiate between the directions or positions of the vertices, but they can distinguish between occupied and unoccupied vertices.
That means each automaton has no compass. 
This limitation in their capabilities makes their behavior on the lattice more complicated.
For example, while an automaton equipped with a compass can easily maintain its movement direction on the lattice, an automaton without a compass would need additional equipment and the development of methods for its use.

It has been proven that a collective consisting of one active automaton and three passive automata, or pebbles, can maintain the direction of movement on a one-dimensional integer lattice. 
On the other hand, no collective consisting of one automaton and less than three pebbles can do so \cite{Kurganskyy2016}.
This paper presents sufficient and necessary conditions in the form of constraints on automata properties and collective structure, that enable the collective of automata to act as a uniform entity and preserve the direction of movement on an infinite square lattice with a width of $2$.

\section{Basic definitions}

Let $\mathbb{Z}$ denote the set of integers, and let $\mathbb{N}$ denote the the set of natural numbers (along with zero).
We will use the symbol $\mathbb{Z}_{m}$ to denote the set $\left\{ 0, 1, ..., m - 1 \right\}$ for any $n \in \mathbb{N}$.
An arbitrary set $E \subseteq \mathbb{Z}^{n}$ is called an $n$-dimensional geometric environment.
The elements of the set $E$ are called the vertices of the environment.
Two vertices $v = \left( a_{1}, a_{2}, ..., a_{n} \right)$ and $v^{\prime} = \left( b_{1}, b_{2}, ..., b_{n} \right)$ are called neighbours if $\vert a_{i} - b_{i} \vert$ for one $i \in \left\{ 1, ..., n \right\}$ and $a_{i} = b_{i}$ for the remaining $i$.
The set of all vertices neighbouring $v \in E$ is called the neighbourhood of $v$.

The direction in the environment $\mathbb{Z}^{n}$ relative to the current vertex is any tuple of $n$ setwise coprime integers.
We say that a vertex $v^{\prime} = \left( b_{1}, b_{2}, ..., b_{n} \right)$ is in the direction $\left( d_{1}, d_{2}, ..., d_{n} \right)$ from a vertex $v = \left( a_{1}, a_{2}, ..., a_{n} \right)$ if $b_{1} = a_{1} + kd_{1}$, $b_{2} = a_{2} + kd_{2}$, ..., $b_{n} = a_{n} + kd_{n}$,
where $k \in \mathbb{N}$.

Suppose a finite automaton $A$ moves in environment $E$. 
Its input is information about current vertex and neighbourhood. 
The automaton's output is to move to a vertex neighbouring the current one, chosen on the basis of the input analysis.
If an automaton $A$ distinguishes between vertices in the current neighbourhood by the coordinates of directions in the environment $E$, we call it an automaton with a compass. 
Otherwise, if it does not use the coordinate system, we call it a compassless automaton.

We will also consider a collective of automata $\mathcal{A} = \left( A_{1}, \dots, A_{m} \right)$ interacting in the environment $E$.
Each automaton $A_{i}$ receives information about the current vertex and its neighbourhood along with information about the presence of other automata of the collective $\mathcal{A}$.
$\mathcal{A}$ is said to be a collective of compassless automata if every automaton in $\mathcal{A}$ is a compassless automaton. 
Below we consider only such collectives of automata.

Let $J \subset \{ 1, \dots, m \}$.
A subsystem $\left( A_{j} \right)_{j \in J}$ of a collective $\mathcal{A} = \left( A_{1}, \dots, A_{m} \right)$ is called automata-pebbles (or pebbles) in this collective if for all $j \in J$ the following conditions hold: (1) $A_{j}$ has a single internal state; (2) $A_{j}$ can only move if there is an automaton $A_{i}$, $i \not\in J$, on the same vertex, and $A_{j}$ can only move to the same vertex as $A_{i}$.
The pebbles play the role of an external memory for automata.
The collective of type $(1, m)$ is called the collective $\mathcal{A} = \left( A_{1}, A_{2}, \dots, A_{m+1} \right)$ consisting of an automaton $A_1$ and $m$ pebbles $A_{2}$, ..., $A_{m+1}$.

We will provide more precise definitions.
The geometric environment $E = \mathbb{Z} \times \mathbb{Z}_{2}$ is called a square lattice of width~$2$.
Let $\mathcal{P}(M)$ be the set of all subsets of any set $M$.
Let $M$ denote the set of all members of the collective $\mathcal{A}$.

Each automaton $A_{i} \in M$, $1 \leq i \leq m+1$, is a sextuple $A_{i} = \left( Q_{i}, X_{i}, Y_{i}, q^{i}_{0}, \varphi_{i}, \psi_{i} \right)$, where 
$Q_{i}$ is a finite set of internal states; 
$X_{i} = \left\{ \left( \alpha, \{ \beta, \gamma, \delta \} \right) \vert\, \alpha, \beta, \gamma, \delta \subseteq M \right\}$ is a finite input alphabet (here $\alpha$ describes the automata in the current vertex, and $\{ \beta, \gamma, \delta \}$ describes the automata in its neighbourhood);
$Y_{i} = \mathcal{P}(M) \cup \{ \mathrm{stay} \}$ is a finite output alphabet (here $y = \{ \mathrm{stay} \}$ means ''stay at the current vertex'', $y = \varnothing$ means ''move to any vertex from the current vertex's neighbourhood that has no automata'', $y \in Y_{i} \setminus \{ \varnothing, \mathrm{stay} \}$ means ''move to the vertex containing only automata from the set $y$'');
$q^{0}_{i} \in Q_{i}$ is an initial state; 
$\varphi_{i} : Q_{i} \times X_{i} \to Q_{i}$ is a transition function;
$\psi_{i} : Q_{i} \times X_{i} \to Y_{i}$ is an output function.

For any pebble $A_{j}$, $2 \leq j \leq m+1$, the following conditions are true:

1) $Q_{j} = \left\{ q^{0}_{j} \right\}$;

2) for any $x = \left( \alpha, \{ \beta, \gamma, \delta \} \right) \in X_{j}$ either $\psi_{j} \left( q^{i}_{0}, x \right) = \mathrm{stay}$, or if $\psi_{j} \left( q^{i}_{0}, x \right) = y \neq \mathrm{stay}$ then $A_{1} \in \alpha$ and there exists $q \in Q_{1}$ such that $\psi_{1} \left( q, x \right) = y$.

We assume that the automaton can distinguish between all pebbles.

The behaviour (nondeterministic) of a collective $\mathcal{A} = \left( A_{1}, A_{2}, ..., A_{m+1} \right)$ of the type $(1, m)$ on the geometric environment $E$ is the set $\Pi (\mathcal{A}, E)$ of sequences $\pi (\mathcal{A}, E)$:
$\left( \vec x_{0}, \vec q_{0}, \vec y_{0} \right)$, ..., $\left( \vec x_{t}, \vec q_{t}, \vec y_{t} \right)$, 
$\left( \vec x_{t+1}, \vec q_{t+1}, \vec y_{t+1} \right)$, ...,
where 
$\vec x_{t} = \left( x_{t}^{1}, \dots, x_{t}^{m+1}\right)$, 
$x_{t}^{j} = \left( \alpha_{t}^{j}, \left\{ \beta_{t}^{j}, \gamma_{t}^{j}, \delta_{t}^{j} \right\} \right) \in X_{j}$,
$\vec q_{t} = \left( q_{t}^{1}, \dots, q_{t}^{m+1} \right)$, 
$q_{t}^{j} \in Q_{j}$,
$\vec y_{t} = \left( y_{t}^{1}, \dots, y_{t}^{m+1} \right)$, $y_{t}^{j} \in Y_{j}$,
$1 \leq j \leq m+1$,
such that
$q_{t+1}^{j} = \varphi_{j} \left( q_{t}^{j}, x_{t}^{j} \right)$, $y_{t+1}^{j} = \psi_{j} \left( q_{t}^{j}, x_{t}^{j} \right)$.
A single sequence $\pi (\mathcal{A}, E)$ is called an implementation of behaviour $\Pi (\mathcal{A}, E)$.

Let the automaton $A_{i} \in M$, $1 \leq i \leq m+1$, be at vertex $v_{i}(t) = \left\{ a_{i_{1}}, a_{i_{2}} \right\}$ at time $t$.
The vector  $v_{\mathcal{A}} (t) = \left( v_{1} (t) + \ldots + v_{m+1} (t) \right) / (m+1)$ is called the coordinate of collective $\mathcal{A}$ at time $t$.
We will call $d_{\mathcal{A}} = \max\left\{ \left| a_{i_{k}} - a_{i_{l}} \right| \,\big\vert\, 1 \leqslant i \leqslant n, 1 \leqslant k, l \leqslant m+1 \right\}$ the diameter of this collective.
The movement of the collective $\mathcal{A}$ in the environment $E$ during the implementation of behaviour $\pi (\mathcal{A}, E)$ is called directed if there are constants $c_{1}, c_{2} \in \mathbb{N}$ such that the diameter $d_{\mathcal{A}} \leq c_{1}$ and for any moment $t$ there are natural $t^{\prime}, t^{\prime\prime} \leq c_{2}$ for which the equality $v_{\mathcal{A}} \left( t+t^{\prime} \right) - v_{\mathcal{A}} (t) = v_{\mathcal{A}} \left( t + t^{\prime} + t^{\prime\prime} \right) - v_{\mathcal{A}} \left( t+t^{\prime} \right)$ holds.
The directional movement of collective $\mathcal{A}$ is called uniform if $t^{\prime} = t^{\prime\prime} = c_{2}$.
In the following, we assume that the movement of a collective $\mathcal{A}$ in the environment $E$ is directed if it is directed for all implementations of its behaviour $\Pi (\mathcal{A}, E)$. 
Otherwise, we say that the movement of this collective is not directed.

A vertex of the environment is called free if it contains no pebbles. 
Otherwise the vertex is called occupied.

Below, pebble configuration will be considered as their relative position on the lattice vertices.
The initial position of the pebbles is called their initial configuration.
We will say that the configurations $K^{\prime}$ and $K^{\prime\prime}$ of the pebbles of collective $\mathcal{A}$ are indistinguishable if any implementation of the behaviour of this collective with the initial configuration $K^{\prime}$ is the implementation of the behaviour of the same collective with the initial configuration $K^{\prime\prime}$ and vice versa.
In the following, we say that two configurations of the pebbles of collective $\mathcal{A}$ are indistinguishable in the worst case if there is at least one implementation of the behaviour of this collective that is the same in both configurations.

A schema $\widetilde K$ of a configuration of pebbles $K$ is a fragment of the environment, where the vertices on which the pebbles are located are marked in some way. 
Obviously, different configurations can have the same schema.
Let's say that each specific configuration of pebbles is an interpretation of the corresponding schema.
Schemes $\widetilde K^{\prime}$ and $\widetilde K^{\prime\prime}$ are said to be indistinguishable if for any interpretation of scheme $\widetilde K^{\prime}$ there is an indistinguishable interpretation of scheme $\widetilde K^{\prime\prime}$ and vice versa. 
We define the worst-case indistinguishability of schemas in the same way as we did for configurations.
An elementary transfer of a pebble is the relocation of one pebble from its current vertex to a neighbouring vertex. 
The movements of the pebbles of collective $\mathcal{A}$ can be represented as a sequence of their elementary transfers and, therefore, as a sequence of configurations or schemes.

Let sub-collective $\mathcal{A}^{\prime}$ of collective $\mathcal{A}$ be an arbitrary subset of the members of that collective.
We say that sub-collective $\mathcal{A}^{\prime}$ is isolated in collective $\mathcal{A}$ if there exists an implementation of the behaviour of that collective in which no member of $\mathcal{A}^{\prime}$ observes any member of $\mathcal{A}$ not belonging to $\mathcal{A}^{\prime}$.
It is obvious that from the isolation of some sub-collective $\mathcal{A}^{\prime}$ of collective $\mathcal{A}$ follows the isolation of the sub-collective $\mathcal{A} \setminus \mathcal{A}^{\prime}$.

\section{Directional movement on a square lattice of width 2}

\begin{theorem}
The movement of the following collectives on a square lattice of width 2 is not directed:
\begin{enumerate}
\item a collective of a single automaton;
\item a collective of an automaton and a pebble;
\item a collective of one automaton and two pebbles;
\item a collective of one automaton and three pebbles.
\end{enumerate}
\end{theorem}

\begin{proof}
1. By definition, a single automaton does not distinguish between free vertices in the neighbourhood of the current vertex. 
Therefore, at each transfer, it chooses a target vertex non-deterministically.
Then there is an implementation of the automaton's behaviour such that its movement is not directed.
For example, an implementation in which the automaton moves only between two neighbouring vertices.
It follows that the movement of a collective containing an isolated sub-collective of a single automaton is also not directed.

2. Let $\mathcal{A} = \left( A_{1}, A_{2} \right)$, which means that the collective is of the type $(1, 1)$.
If the automaton moves away from the pebble by a distance greater than $1$, it will form an isolated sub-collective of the collective $\mathcal{A}$.
Since the movement of a single automaton is not directed, the movement of the entire collective in this case is also not directed.
Thus, the automaton either moves with the pebble or moves in the neighbourhood of the vertex occupied by the pebble.
It is clear that any individual movements of the automaton in the neighbourhood of the vertex with the pebble do not lead to a directed movement of the collective $\mathcal{A}$.
If the automaton moves along with the pebble, then due to the non-deterministic choice of the target vertex, there exists an implementation of the behaviour of the collective $\mathcal{A}$, in which it moves between two neighbouring vertices.
Thus, the movement of a collective of one automaton and one pebble is not directed. 
Also, the movement of a collective containing an isolated sub-collective of one automaton and one pebble is not directed.

3. Let $\mathcal{A} = \left( A_{1}, A_{2}, A_{3} \right)$, which means that the collective is of the type $(1, 2)$.
It is important to ensure that the automaton and pebbles are positioned either on the same vertex or on pairwise neighbouring vertices to prevent the formation of isolated sub-collectives of a single automaton or an automation and one of the pebbles. 
This is necessary as the formation of such sub-collectives makes it impossible for the entire collective to perform directed movement.
From the above, it can be deduced that all potential configurations of the pebbles of the collective A can be depicted using the five schemes presented in Figure \ref{fig:fig1}.

\begin{figure}
	\centering
	\includegraphics[scale=0.7, clip]{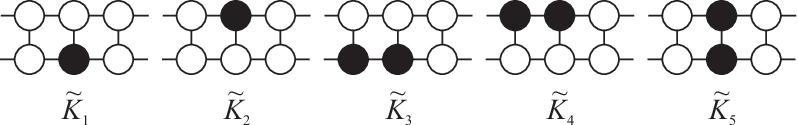}
	\caption{Schemes of configurations of pebbles from a collective of type $(1, 2)$.}
	\label{fig:fig1}
\end{figure}

Schemes $\widetilde K_{1}$ and $\widetilde K_{2}$ are indistinguishable due to the symmetry of the environment $E$.
Schemes $\widetilde K_{3}$ -- $\widetilde K_{5}$ are also indistinguishable for the same reason.
Indeed, automaton $A_{1}$, when on any of the pebbles in any of the interpretations of these schemes, observes two free vertices and one vertex occupied by another pebble in the neighbourhood of the current vertex.
The observations that the automaton can obtain during any individual movements not prohibited by the previous condition do not enable it to distinguish between the interpretations of the above schemes.
If the automaton transfers one of the pebbles to a vertex occupied by another pebble, the current interpretation of one of the schemes $\widetilde K_{3}$ -- $\widetilde K_{5}$ is transformed into the corresponding interpretation of schemes $\widetilde K_{1}$ or $\widetilde K_{2}$.
Once the automaton and the pebbles have arrived at the same vertex, they are unable to distinguish between free vertices in its neighbourhood by definition.
Then, for any initial configuration of pebbles, there exists an implementation of the behaviour of the collective $\mathcal{A}$ in which its movement is not directed. 
For example, an implementation in which it moves only in the neighbourhood of the initial vertex. 
Additionally, the movement of a collective containing a isolated sub-collective of one automaton and two pebbles is not directed.

\begin{figure}[h]
	\centering
	\includegraphics[scale=0.7, clip]{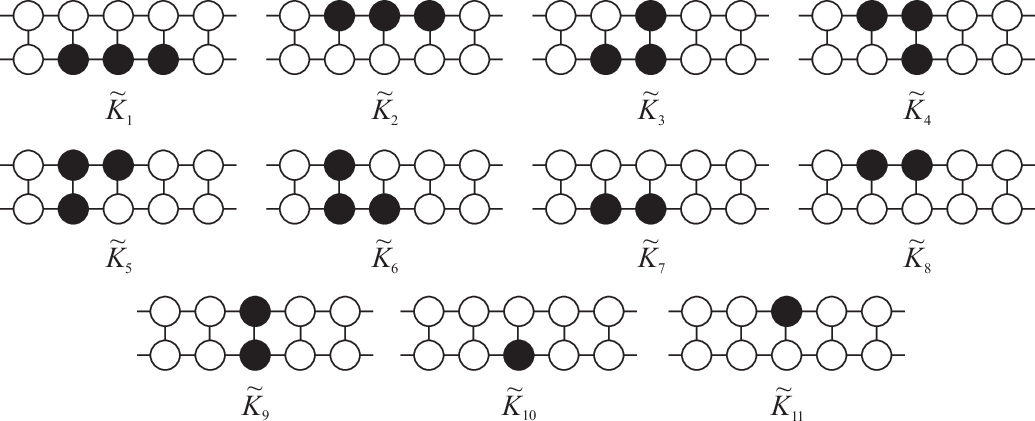}
	\caption{Schemes of configurations of pebbles from a collective of type $(1, 3)$.}
	\label{fig:fig2}
\end{figure}

4. Let $\mathcal{A} = \left( A_{1}, A_{2}, A_{3}, A_{4} \right)$, which means that the collective is of the type $(1, 3)$.
Since none of the automata in the collective can distinguish between free vertices in the neighbourhood of its current vertex, the location of collective members on non-neighbouring vertices leads to the formation of isolated sub-collectives of one automaton, an automaton and one of the pebbles, or an automaton and any two pebbles.
As the formation of such sub-collectives makes it impossible for the entire collective to perform directed movement, we will assume that the members of the collective are located on pairwise neighbouring vertices.
It follows that all potential configurations of the pebbles of the collective A can be represented in the form of eleven schemes shown in Figure \ref{fig:fig2}.

Suppose collective $\mathcal{A}$ is performing directed movement in environment $E$.
Let us also assume that the interpretations of schemes $\widetilde K_{10}$ or $\widetilde K_{11}$ occur in the sequence of pebble configurations that arises when this collective moves.
The information about the previous position of the pebbles, and therefore the direction of movement, is irrevocably lost if all the pebbles end up on the same vertex.
Then there exists an implementation of the behaviour of collective $\mathcal{A}$ in which its movement is not directed. 
Thus the interpretations of schemes $\widetilde K_{10}$ or $\widetilde K_{11}$ can appear in the sequence of configurations only once -- at the first position.
It is easy to check that, due to the symmetry of the square lattice, the following sets consist of pairwise indistinguishable schemes: 
$\left\{ \widetilde K_{3}, \widetilde K_{4}, \widetilde K_{5}, \widetilde K_{6} \right\}$, 
$\left\{ \widetilde K_{1}, \widetilde K_{2} \right\}$,
$\left\{ \widetilde K_{7}, \widetilde K_{8} \right\}$, and $\left\{ \widetilde K_{10}, \widetilde K_{11} \right\}$.

We will show that schemes $\widetilde K_{1}$ and $\widetilde K_{2}$ are indistinguishable in the worst case from schemes $\widetilde K_{3}$, $\widetilde K_{4}$, $\widetilde K_{5}$, and $\widetilde K_{6}$.
In each of these schemes, each pebble is placed on a separate vertex.
In each of these schemes, it is forbidden to move a pebble located on a vertex neighbouring two occupied vertices in order to prevent the formation of isolated sub-collectives.
The pebble placed on this vertex is called the middle pebble, and the other two pebbles are called the outer pebbles.
It is evident that two interpretations of the same scheme with different middle pebbles are distinguishable.
Let us choose all the interpretations of the schemes being considered, in which the middle pebble is the same.
Among the chosen interpretations, we will randomly select interpretation $K_{i}$ of scheme $\widetilde K_{1}$ and interpretation $K_{j}$ of scheme $\widetilde K_{3}$.
It is clear that the automaton, situated on the same pebble in configuration $K_{i}$ and $K_{j}$, receives identical input information.

Suppose that the pebbles are in configuration $K_{i}$. 
Then the automaton can move from the vertex with the middle pebble to the only free vertex in its neighbourhood. 
If the automaton moves from this vertex to any of the two free vertices neighbouring it, the only occupied vertex in the neighbourhood of the current vertex is the vertex with one of the outer pebbles.
If the pebbles were in the configuration $K_{j}$, the worst case scenario is that the same sequence of moves would lead to the formation of an isolated sub-collective consisting of the automaton $A_{1}$. 
This would make any further directed movement of the collective $\mathcal{A}$ impossible.
Thus, this sequence of moves is not permissible when the configuration of the pebbles is unknown a priori.

Suppose further that the pebbles are in the configuration $K_{j}$.
In this case, there is a free vertex that is neighbouring each of the vertices with the outer pebbles. 
Note that in configuration $K_{i}$ there is no free vertex with the above property.
As there are two free vertices neighbouring each vertex with an outer pebble, but only one satisfies the aforementioned property, the automaton, by definition, is unable to determine the desired vertex uniquely.

Thus, prior to the pebbles starting to transfer, the automaton at its worst cannot distinguish between the configurations of $K_{i}$ and $K_{j}$.

Let us examine the changes of the configurations $K_{i}$ and $K_{j}$ during elementary transfers of the same pebbles.
The configuration $K_{i}$ is transformed into one of the interpretations of the scheme $\widetilde K_{7}$ by the elementary transfer of any of the outer pebbles. 
In its turn, the configuration $K_{j}$ is transformed either into the interpretation of scheme $\widetilde K_{7}$ or into the interpretation of scheme $\widetilde K_{9}$.
Suppose that configuration $K_{i}$ has been transformed into interpretation $K^{\prime}_{i}$ of scheme $\widetilde K_{7}$, and configuration $K_{j}$ has been transformed into interpretation $K^{\prime}_{j}$ of scheme $\widetilde K_{9}$.
It is clear that the automaton, being on the vertices of configurations $K^{\prime}_{i}$ and $K^{\prime}_{j}$ occupied by the same pebbles, obtains identical information as input and cannot distinguishing between these configurations.
With the next elementary transfer of the pebble to a free vertex, the configuration $K^{\prime}_{i}$ is transformed into one of the interpretations of the schemes $\widetilde K_{3}$, $\widetilde K_{6}$ or $\widetilde K_{1}$. In its turn, the configuration $K^{\prime}_{j}$ is transformed into one of the interpretations of the schemes $\widetilde K_{3}$, $\widetilde K_{4}$, $\widetilde K_{5}$ or $\widetilde K_{6}$.

Thus, there exists an implementation of the behaviour of collective $\mathcal{A}$ that is the same for the initial configurations $K_{i}$ and $K_{j}$. 
Hence, these configurations are indistinguishable in the worst case. 
As these configurations were randomly chosen, it follows that schemes $\widetilde K_{1}$--$\widetilde K_{6}$ are indistinguishable in the worst case.
Note that the above also implies that schemes $\widetilde K_{7}$, $\widetilde K_{8}$, and $\widetilde K_{9}$ are indistinguishable in the worst case.

\begin{figure}[h]
	\centering
	\includegraphics[scale=0.9, clip]{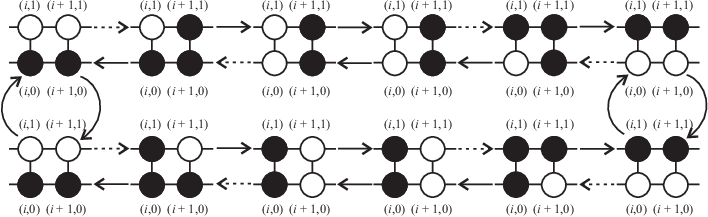}
	\caption{Transformation of the schemes of configurations of pebbles for collective type $(1, 3)$ in the worst case.}
	\label{fig:fig3}
\end{figure}

The example in Figure \ref{fig:fig3} shows a diagram of the transformation of schemes in which the pebbles in the worst case do not leave some limited area of the lattice. 
Here, a solid arrow denotes the transfer of a pebble to an occupied vertex, while a dashed arrow represents the transfer to a free vertex.

Therefore, there exists an implementation of the behaviour of collective $\mathcal{A}$ of type $(1, 3)$, in which its movement is not directed.
\qed
\end{proof}

\begin{theorem}
There exists a collective of one automaton and four pebbles that perform a directed movement on a square lattice of width $2$.
\end{theorem}

\begin{proof}
Let us construct a collective $\mathcal{A} = \left( A_{1}, A_{2}, A_{3}, A_{4}, A_{5} \right)$ that performs directed movement in the environment $E$.
Let us define variables $B$, $C$, $D$, and $H$ which have values in the set of pebbles $\left( A_{2}, A_{3}, A_{4}, A_{5} \right)$, where $B \neq C \neq D \neq H$.
Suppose that initially the pebbles are placed as shown in Figure \ref{fig:fig4}(a) and the automaton $A_{1}$ is located on the same vertex as the pebble $B$. 
Assume that the direction in which the collective A should move is determined by the initial position of the pebbles $B$, $C$ and $D$. 
The collective's strategy is to follow Algorithm 1.

\begin{algorithm}
\caption{Directed movement of collective $\mathcal{A}$ of type (1,4)}
\begin{algorithmic}[1]
\Require initial coordinate $v_{\mathcal{A}}(0)$
\Ensure a sequence of coordinates $\left( v_{\mathcal{A}}(0), v_{\mathcal{A}}(1), \dots \right)$ that satisfies the definition of directed movement
\Statex
\Loop
\State $A_{1}$ and $B$ move to the vertex where $C$ is located
\State $A_{1}$ and $C$ move to the vertex where $D$ is located
\State $A_{1}$ and $D$ move to a free vertex
\If{there is $H$ in the neighbourhood of the current vertex}
\State $A_{1}$ moves to the vertex where $H$ is located
\State $A_{1}$ and $H$ move to the vertex where $D$ is located
\State $A_{1}$ and $D$ move to the vertex where $C$ is located
\State $A_{1}$ and $D$ move to a free vertex
\Else
\State $A_{1}$ moves to the vertex where $C$ is located
\State $A_{1}$ moves to the vertex where $B$ is located
\State $A_{1}$ moves to the vertex where $H$ is located
\State $A_{1}$ and $H$ move to the vertex where $B$ is located
\State $A_{1}$ and $H$ move to the vertex where $C$ is located
\State $A_{1}$ and $H$ move to a free vertex
\EndIf
\State $A_{1}$ moves to the vertex where $C$ is located
\State $A_{1}$ moves to the vertex where $B$ is located
\EndLoop
\end{algorithmic}
\end{algorithm}

Let us demonstrate that this strategy does result in the directed movement of the collective $\mathcal{A}$. 
Figures \ref{fig:fig4}--\ref{fig:fig6} illustrate the movement of the collective while executing the Algorithm 1.

\begin{figure}[h]
	\centering
	\includegraphics[scale=0.8, clip]{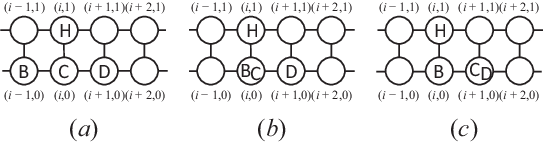}
	\caption{Execution of Algorithm 1 (lines 2--3)}
	\label{fig:fig4}
\end{figure}

Let us recall the assumption that at the initial moment of time the pebbles of collective $\mathcal{A}$ are placed as shown in Figure \ref{fig:fig4}(a), and the automaton $A_{1}$ is located on the same vertex as the pebble $B$.
Since each pebble is unique, the algorithm's instructions in lines 2 and 3 are executed deterministically. 
The positions of the pebbles after each of these lines are shown in Figure \ref{fig:fig4}(b) and Figure \ref{fig:fig4}(c), respectively. 
When executing the instruction in line 4, the automaton randomly selects one of the two free vertices in the neighbourhood of the current vertex and moves to it with the pebble $D$.

\begin{figure}[h]
	\centering
	\includegraphics[scale=0.8, clip]{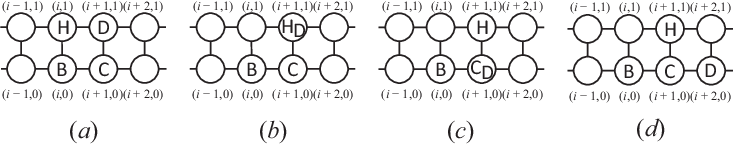}
	\caption{Execution of Algorithm 1 (lines 4, 6--9)}
	\label{fig:fig5}
\end{figure}

If the new current vertex is vertex $(i+1, 1)$ (see Figure \ref{fig:fig5}(a)), then $A_{1}$ observes the pebble $H$ in its neighbourhood. 
In this case, $A_{1}$ leaves $D$ in the current vertex and deterministically moves to the vertex where $H$ is located (line 6). 
Then $A_{1}$ and $H$ move deterministically to the vertex where $D$ is located (line 7 and Figure \ref{fig:fig5}(b)). 
After that, $A_{1}$ and $D$ also move deterministically to the vertex where $C$ is located (line 8 and Figure \ref{fig:fig5}(c)). 
In the neighbourhood of the current vertex, there is a single free vertex ($(i+2, 0)$ in Figure \ref{fig:fig5}(c)), to which $A_{1}$ and $D$ move deterministically (line 9 and Figure \ref{fig:fig5}(d)).

\begin{figure}
	\centering
	\includegraphics[scale=0.8, clip]{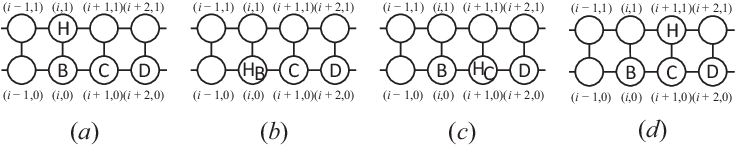}
	\caption{Execution of Algorithm 1 (lines 4, 11--16)}
	\label{fig:fig6}
\end{figure}

If, after line 4, $A_{1}$ and $D$ are at vertex $(i+1, 0)$, then there is no pebble $H$ in its neighbourhood (see Figure \ref{fig:fig6}(a)).
Note that in this scenario, pebbles $B$, $C$, and $D$ have already been moved to their respective vertices, and all that remains is to position pebble $H$ at its new vertex.
The automaton $A_{1}$ moves deterministically through the vertices with pebbles $C$ and $B$ to the vertex with $H$ (lines 11-13). 
Then, $A_{1}$ and H move deterministically to the vertex with pebble $B$ through the vertex with $C$ (lines 14, 15 and Figures \ref{fig:fig6}(b) and \ref{fig:fig6}(c)). 
In the neighbourhood of the current vertex, there is a single free vertex ($(i+1, 1)$ in Figure \ref{fig:fig6}(c)), to which $A_{1}$ and H move deterministically (line 16 and Figure \ref{fig:fig6}(d)).

Automaton $A_{1}$ moves deterministically from vertex with a pebble $B$ to vertex with a pebble $C$ at the end of the current iteration of the algorithm.

Suppose that at time $t=0$, collective $\mathcal{A}$ is located at vertices $( i-1, 0 )$, $( i, 0 )$, $( i+1, 0 )$, and $( i, 1 )$. 
Then its diameter is $2$ and its coordinate is $\left( i-\frac{1}{5}, \frac{1}{5} \right)$. 
At the end of the first iteration of the main algorithm cycle, collective $\mathcal{A}$ has moved in the direction $( 1, 0 )$ and is located at vertices $( i, 0 )$, $( i+1, 0 )$, $( i+2, 0 )$, and $( i+1, 1 )$, its diameter remains $2$, and its coordinate is $\left( i+\frac{1}{5}, \frac{1}{5} \right)$.
One iteration of the algorithm's main cycle can take either 9 or 11 steps, depending on the non-deterministic transition to a free vertex in line 4. 
Let $t^{\prime}$ indicate the duration of the first iteration.
At the end of the second iteration, collective $\mathcal{A}$ has moved again in the direction $( 1, 0 )$ and is now positioned at vertices $( i+1, 0 )$, $( i+2, 0 )$, $( i+3, 0 )$, and $( i+2, 1 )$. 
Its diameter remains $2$, and its coordinate is $\left( i+\frac{9}{5}, \frac{1}{5} \right)$.
Let $t^{\prime\prime}$ represent the time required to complete the second iteration.
It is easy to see that $v_{\mathcal{A}} \left( t^{\prime} \right) - v_{\mathcal{A}} (0) = v_{\mathcal{A}} \left( t^{\prime}+t^{\prime\prime} \right) - v_{\mathcal{A}} \left( t^{\prime} \right) = (1, 0)$.
Suppose that at the end of the $k$-th iteration at time $t$, collective $\mathcal{A}$ is located at vertices $( i+k-1, 0 )$, $( i+k, 0 )$, $( i+k+1, 0 )$, and $( i+k, 1 )$.
Then $v_{\mathcal{A}}(t) = \left( i+k+\frac{1}{5}, \frac{1}{5} \right)$ and the diameter is $2$.
After the $(k+1)$-th iteration, which took $t^{\prime} \in \{ 9, 11 \}$ steps, collective $\mathcal{A}$ moved in the direction $(1, 0)$ and is at vertices $( i+k, 0 )$, $( i+k+1, 0 )$, $( i+k+2, 0 )$, and $( i+k+1, 1 )$. 
In this case, $v_{\mathcal{A}} \left( t+t^{\prime} \right) = \left( t+k+\frac{4}{5}, \frac{1}{5} \right)$, and the diameter is still $2$.
After the $(k+2)$-th iteration, which took $t^{\prime\prime} \in \{ 9, 11 \}$ steps, collective $\mathcal{A}$ moved in the direction $(1, 0)$ and is at vertices $( i+k+1, 0 )$, $( i+k+2, 0 )$, $( i+k+3, 0 )$, and $( i+k+2, 1 )$.
Its diameter remains $2$, and $v_{\mathcal{A}} \left( t+t^{\prime}+t^{\prime\prime} \right) = \left( i+k+\frac{9}{5}, \frac{1}{5} \right)$.
It is easy to see that $v_{\mathcal{A}} \left( t+t^{\prime} \right) - v_{\mathcal{A}} (t) = v_{\mathcal{A}} \left( t+t^{\prime}+t^{\prime\prime}\right) - v_{\mathcal{A}} \left( t+t^{\prime} \right) = ( 1, 0 )$.
Thus, the movement of collective $\mathcal{A}$ is directed. \qed
\end{proof}

\paragraph{Conclusion}
In this paper, we present the necessary and sufficient conditions for a collective consisting of an automaton and a finite set of pebbles to maintain its direction of movement in an environment that is an infinite square lattice of width 2.
This study establishes a foundation for future investigations into the behavior of automaton collectives without a compass in discrete geometric environments.

\paragraph{Acknowledgements}
The author is partially supported by the Grant EFDS-FL2-08 of the found The European Federation of Academies of Sciences and Humanities (ALLEA)

\bibliographystyle{vancouver}
\bibliography{references}  

\begin{thebibliography}{1}

\bibitem{Hemmerling1989}
Hemmerling A.
\newblock Labyrinth Problems.
\newblock BSB Teubner Verlag; 1989.
\newblock doi:10.1007/978-3-322-94560-0.

\bibitem{Kilibarda2003}
Kilibarda G, Kudryavtsev VB, U{\v{s}}{\'{c}}umli{\'{c}} {\v{S}}.
\newblock Independent systems of automata in labyrinths.
\newblock Discrete Mathematics and Applications. 2003;13(3):221-55.
\newblock doi:10.1515/156939203322385847.

\bibitem{Kilibarda2003a}
Kilibarda G, Kudryavtsev VB, U{\v{s}}{\'{c}}umli{\'{c}} {\v{S}}.
\newblock Collectives of automata in labyrinths.
\newblock Discrete Mathematics and Applications. 2003;13(5):429-66.
\newblock doi:10.1515/156939203322694736.

\bibitem{Kari2006}
Kari J.
\newblock Image Processing Using Finite Automata.
\newblock In: Esik~Z MV Martin-Vide~C, editor. Recent Advances in Formal
  Languages and Applications. Studies in Computational Intelligence. Springer
  Berlin Heidelberg; 2006. p. 171-208.
\newblock doi:10.1007/978-3-540-33461-3{\_}7.

\bibitem{Stamatovic2017}
Stamatovic B, Kilibarda G.
\newblock Algorithm for Identification of Infinite Clusters Based on Minimal
  Finite Automaton.
\newblock Mathematical Problems in Engineering. 2017;2017:1-7.
\newblock Article ID 8251305.
\newblock doi:10.1155/2017/8251305.

\bibitem{Dudek2010}
Dudek G, Jenkin M.
\newblock Computational Principles of Mobile Robotics.
\newblock Cambridge University Press; 2010.
\newblock doi:10.1017/cbo9780511780929.

\bibitem{Blum1978}
Blum M, Kozen D.
\newblock On the power of the compass (or, why mazes are easier to search than
  graphs).
\newblock In: 19th Annual Symposium on Foundations of Computer Science (sfcs
  1978). {IEEE}; 1978. doi:10.1109/sfcs.1978.30.

\bibitem{Donald2012}
Donald B.
\newblock The Compass That Steered Robotics.
\newblock In: Constable R, Silva A, editors. Logic and Program Semantics.
  Springer Berlin Heidelberg; 2012. p. 50-65.
\newblock doi:10.1007/978-3-642-29485-3{\_}5.

\bibitem{Kurganskyy2016}
Kurganskyy OM, Sapunov SV.
\newblock On the Directional Movement of a Collective of Automata without a
  Compass on a One-dimensional Integer Lattice.
\newblock Izvestiya of Saratov University New Series Series: Mathematics
  Mechanics Informatics. 2016;16(3):356-65.
\newblock (in Russian).
\newblock doi:10.18500/1816-9791-2016-16-3-356-365.

\end{thebibliography}

\end{document}